\documentclass{article}
\usepackage{graphicx} 
\usepackage{fullpage}
\usepackage{braket}
\usepackage{amsfonts}
\usepackage{authblk}
\usepackage{tikz}
\usetikzlibrary{positioning}
\usetikzlibrary{fit}
\usepackage{amsmath}
\usepackage{amsthm}
\usepackage{appendix}

\newtheorem{theorem}{Theorem}
\theoremstyle{remark}
\newtheorem*{remark}{Remark}
\theoremstyle{definition}
\newtheorem{definition}{Definition}
\newtheorem{example}{Example}

\title{Explicit C*-algebraic Protocol for Exact Universal Embezzlement of Entanglement}
\author{Li Liu}
\affil{University of Copenhagen}
\date{}

\begin{document}

\maketitle
\begin{abstract}
We present an explicit construction of a universal embezzlement protocol in the C*-algebraic model of quantum information, that is equivalent to the commuting operator model. Our protocol enables exact embezzlement of arbitrary bipartite pure states using a single, fixed catalyst state. Unlike prior constructions that achieve only approximate embezzlement or require state-dependent catalysts, our approach is both exact and state-independent. The construction is explicit, based on simple *-automorphisms acting locally on infinite tensor products of CAR algebras with the underlying idea of the Hilbert hotel. In the dense-state case, the protocol naturally recovers the Type III$_1$ factor via the GNS construction, consistent with recent classification results\cite{van2024embezzlement}. We further extend the construction to allow exact embezzlement of all states, at the cost of working with a non-separable C*-algebra. Despite the increase in algebraic size, the operational structure remains simple and localized. This offers a conceptually intuitive model for universal entanglement embezzlement in infinite-dimensional settings.
\end{abstract}
\section{Introduction}

Entanglement is a fundamental resource in quantum information theory that cannot be generated by local operations and classical communication. Embezzlement of entanglement refers to the striking phenomenon where entanglement can be “borrowed” from a large catalyst state and “restored” afterward, enabling the effective creation of entanglement through only local unitaries.  Mathematically, let $\ket\psi \in \mathcal{H}$ be a state shared between Alice and Bob, and let $\ket\phi \in \mathbb{C}^n \otimes \mathbb{C}^n$ be another shared state. An embezzlement protocol implements the mapping
\begin{equation}
\ket\psi \otimes\ket{00} \mapsto \ket{\psi}\otimes\ket\phi
\end{equation}
using only local operations by Alice and Bob. Without loss of generality, we assume that $\ket\phi$ is entangled, as the problem becomes trivial otherwise.

Since its introduction in \cite{van2002embezzling}, embezzlement has played a central role in the study of nonlocal correlations, resource theories, and the structure of quantum strategies. There has been numerous recent studies on characterization of embezzlement protocols \cite{Schwartzman:2024complexity, vanLuijk:2024multipartite, van2024embezzlement,  zanoni2024complete}.

In the original protocol of van Dam and Hayden \cite{van2002embezzling}, approximate embezzlement is achieved using a family of finite-dimensional catalyst states, each tailored to a target precision. The protocol is universal but not exact: for each target state, one obtains an approximation with error vanishing in the limit of growing catalyst dimension. Later work by Cleve et al. \cite{cleve2017perfect} constructed an exact embezzlement protocol, but only for a fixed target state (e.g., a Bell pair) using infinite-dimensional systems in the commuting operator model.

Recent non-constructive results \cite{van2024embezzlement} showed that exact universal embezzlement is possible using a single catalyst state in a Type III$_1$ von Neumann algebra. However, no explicit protocol has been known that achieves this -- either exactly or universally -- using a single, well-defined catalyst.

\vspace{1em}
\noindent\textbf{Our Contribution.} In this work, we construct the \emph{first explicit protocol for exact universal embezzlement} using a single catalyst state. Our protocol is formulated in the C*-algebraic model, which captures locality through tensor products of C*-algebras and is known to be equivalent to the commuting operator model via the GNS construction. 

The construction is exact: it perfectly embezzles any target entangled state (within a dense set, or all states in the non-separable case). It is universal: the same catalyst state enables embezzlement of all target states. And most importantly, it is \emph{explicit}: we describe a concrete, intuitive protocol based on shift and swap operations, inspired by the Hilbert hotel. The resulting catalyst state may appear large — as it is built from infinitely many copies of component states — but the overall complexity of the system remains infinite-dimensional, and the underlying mechanism is clean and structurally transparent. The idea behind the protocol is similar to the ones mentioned in \cite{leung2013coherent}, except their result is in the approximate finite dimensional setting.

We begin by reviewing the C*-algebraic model and the notion of embezzlement in this setting. We then describe our Bell-pair embezzlement protocol using CAR algebras, and extend it to arbitrary 2-qubit and $2n$-qubit states. Next, we construct an explicit protocol for exact embezzlement of a dense set of bipartite entangled states using a single infinite-dimensional catalyst state. This already achieves universal embezzlement in a strong sense and aligns with recent classification results. We then extend our construction to enable exact embezzlement of \emph{all} bipartite pure states. This extension requires moving beyond separable algebras to a non-separable C*-algebraic setting. While this increases the algebraic complexity, it does not alter the underlying intuition or the structure of the protocol, which remains elementary and explicit..

\section{The C*-model and Bell-state Embezzlement}

\subsection{Introduction to the C*-Algebraic Model}

Before delving into the construction of a universal embezzlement protocol, we begin with a fundamental building block: the single-state embezzlement protocol described in \cite{cleve2017perfect}, now translated into the C*-algebraic framework. This translation has been carried out in my Ph.D. thesis \cite{thesis}; however, due to its central role in the overall construction, we restate the key elements here. We start by reviewing the C*-algebraic model for quantum information, also discussed in \cite{cleve2022constant}.

The goal of the C*-algebraic model is to describe quantum systems using C*-algebras rather than Hilbert spaces. In this framework, each quantum system is associated with a C*-algebra. The composition of local quantum systems is represented by the tensor product of their corresponding local C*-algebras, which forms the global C*-algebra of the combined system. States of quantum systems are defined as abstract states on the C*-algebra: given a C*-algebra $\mathcal{A}$, a state is a linear functional $s: \mathcal{A} \to \mathbb{C}$ such that $s \geq 0$ and $s(1) = 1$.

Measurements are modeled by POVMs whose elements are drawn from the C*-algebra, and dynamical transformations (such as unitary gates) are modeled by *-automorphisms of the algebra. Since there are multiple notions of tensor products for C*-algebras, we adopt the \emph{maximal tensor product} $\otimes_{\max}$ throughout, which can be shown to be equivalent to the commuting operator model. Unless otherwise specified, all tensor products should be interpreted as $\otimes_{\max}$.

The following table summarizes the comparison between the Hilbert space model and the C*-algebraic model:

\begin{table}[h]
\centering
\begin{tabular}{|l|l|l|}
\hline
& Hilbert space model & C*-algebraic model \\
\hline\hline
Quantum system & Hilbert space $\mathcal{H}$ & C*-algebra $\mathcal{A}$ \\
\hline
Composition of local systems & $\mathcal{H}_1 \otimes \mathcal{H}_2$ & $\mathcal{A}_1 \otimes \mathcal{A}_2$ \\
\hline
States & Density operators on $\mathcal{H}$ & Unital positive linear functionals on $\mathcal{A}$ \\
\hline
Measurements & POVMs on $\mathcal{H}$ & POVMs with elements from $\mathcal{A}$ \\
\hline
Dynamics & Unitary operators on $\mathcal{H}$ & *-automorphisms of $\mathcal{A}$ \\
\hline
\end{tabular}
\caption{Hilbert space model vs. C*-algebraic model}
\label{tab:cstar_vs_hilbert}
\end{table}

An abstract state $s$ generalizes the notion of a density matrix $\rho$ via the relation:
\begin{equation}
s(A) = \operatorname{Tr}(\rho \pi(A)),
\end{equation}
where $\pi(A)$ is a representation of the element $A \in \mathcal{A}$ as an operator on the Hilbert space in which $\rho$ resides.

As discussed in \cite{cleve2022constant}, *-automorphisms are used to model dynamics because they are the natural generalization of unitary transformations: they preserve both the algebraic structure and the norm of the C*-algebra. While this definition restricts attention to reversible operations, it is also sufficient for general quantum channels: by Stinespring's dilation theorem, any quantum channel can be implemented as a *-automorphism on a larger system.

\subsection{Definition of Embezzlement in the C*-Algebraic Model}

We begin by defining two-qubit entanglement within the C*-algebraic framework. The definition naturally extends to $n$-qubit systems by appropriately generalizing the target state.

Let $\mathcal{A}$ and $\mathcal{B}$ be the C*-algebras associated with Alice and Bob’s local quantum systems, respectively. Let $\ket{\phi} \in \mathbb{C}^2 \otimes \mathbb{C}^2$ be a two-qubit entangled target state. Define the following states on matrix algebras:
\begin{align}
s_0 : \mathbb{M}_2 &\to \mathbb{C}, \quad s_0(M) = \bra{0} M \ket{0}, \\
s_\phi : \mathbb{M}_2 \otimes \mathbb{M}_2 &\to \mathbb{C}, \quad s_\phi(M) = \bra{\phi} M \ket{\phi}.
\end{align}

A \emph{universal embezzlement protocol} in the C*-model consists of a collection $\{s, \alpha_{A,\phi}, \alpha_{B,\phi}\}_{\ket{\phi}}$, where:
\begin{itemize}

\item $s : \mathcal{A} \otimes \mathcal{B} \to \mathbb{C}$ is a fixed state (the catalyst),

\item $\alpha_{A,\phi}$ is a *-automorphism on $\mathcal{A} \otimes \mathbb{M}_2$,

\item $\alpha_{B,\phi}$ is a *-automorphism on $\mathcal{B} \otimes \mathbb{M}_2$,
\end{itemize}

such that for every two-qubit target state $\ket{\phi} \in \mathbb{C}^2 \otimes \mathbb{C}^2$, the following embezzlement condition holds:
\begin{equation}
(s \otimes s_0 \otimes s_0)\left( (\alpha_{A,\phi} \otimes \alpha_{B,\phi})(X) \right) = (s \otimes s_\phi)(X)
\end{equation}
for all $X \in \mathcal{A} \otimes \mathcal{B} \otimes \mathbb{M}_2 \otimes \mathbb{M}_2$.

Here, we suppress the reordering of the algebras required for tensor compatibility in the expression $(\alpha_{A,\phi} \otimes \alpha_{B,\phi})(X)$ for readability; formally, the *-automorphisms act on the respective subsystems involving the ancilla registers in a consistent manner.

We note that applying the GNS construction to our C*-algebraic embezzlement protocol yields a corresponding embezzlement protocol in the commuting operator model. A key detail in this translation is that our protocol relies on local operations implemented as outer *-automorphisms on the C*-algebras associated to Alice and Bob. To make the GNS construction applicable, we must first extend the original C*-algebra by cross product so that these outer *-automorphisms become inner in the extended algebra. The GNS construction can then be applied to the extended algebra and the given state, producing a representation in which the embezzlement protocol is realized via unitary operations.

\subsection{Bell-State Embezzlement Intuition}

Fixing $\ket{\phi}$ to be the Bell state
\begin{equation}
    \ket\phi = \frac{\ket{00} + \ket{11}}{\sqrt 2},
\end{equation}
we get the single-state embezzlement protocol capable of embezzling a Bell state, which is precisely the protocol introduced in \cite{cleve2017perfect}. We now describe this Bell-state embezzlement protocol within the C*-algebraic framework, inspired by the commuting operator model in \cite{cleve2017perfect} and its underlying intuition—the \emph{Hilbert hotel} analogy.

The core idea of the protocol is straightforward: If the catalyst state $\ket{\psi}$ consists of infinitely many Bell pairs along with infinitely many $\ket{00}$ states, one can “shift” all Bell states toward the $\ket{00}$ region, then swap the Bell state at the intersection with an external $\ket{00}$ pair. This process effectively produces a Bell state outside, while leaving the catalyst state unchanged.

This intuition can be visualized in the following sequence of figures.

\begin{figure}[h!]
    \centering
\begin{tikzpicture}[
qubitnode/.style={circle, draw=black!60, fill=white!5, very thick, minimum size=5mm},
]
\node[qubitnode]   (A1)                       {};
\node[qubitnode]   (A2)     [right=of A1]     {};
\node[qubitnode]   (A3)     [right=of A2]     {};
\node[qubitnode]   (A4)     [right=of A3]     {};
\node[qubitnode]   (A5)     [right=of A4]     {};
\node[qubitnode]   (A6)     [right=of A5]     {};
\node[qubitnode]   (A7)     [right=of A6]     {};
\node[qubitnode]   (A8)     [right=of A7]     {};
\node[below right=0.3cm and 0.5cm of A8] {$\cdots$};
\node[below left=0.3cm and 0.5cm of A1] {$\cdots$};
\node[qubitnode]   (B1)     [below=of A1]     {};
\node[qubitnode]   (B2)     [right=of B1]     {};
\node[qubitnode]   (B3)     [right=of B2]     {};
\node[qubitnode]   (B4)     [right=of B3]     {};
\node[qubitnode]   (B5)     [right=of B4]     {};
\node[qubitnode]   (B6)     [right=of B5]     {};
\node[qubitnode]   (B7)     [right=of B6]     {};
\node[qubitnode]   (B8)     [right=of B7]     {};
\node              (0)      [right=of B8]     {$\ $};
\node[]            (A)      [left=of A1]      {Alice};
\node[]            (B)      [left=of B1]      {Bob};

\draw (A5.south) -- (B5.north);
\draw (A6.south) -- (B6.north);
\draw (A7.south) -- (B7.north);
\draw (A8.south) -- (B8.north);

\node[]            (i-3)    [above=0cm of A1]     {-3};
\node[]            (i-2)    [above=0cm of A2]     {-2};
\node[]            (i-1)    [above=0cm of A3]     {-1};
\node[]            (i0)     [above=0cm of A4]     {0};
\node[]            (i1)     [above=0cm of A5]     {1};
\node[]            (i2)     [above=0cm of A6]     {2};
\node[]            (i3)     [above=0cm of A7]     {3};
\node[]            (i4)     [above=0cm of A8]     {4};
\node[]            (j-3)    [below=0cm of B1]     {-3};
\node[]            (j-2)    [below=0cm of B2]     {-2};
\node[]            (j-1)    [below=0cm of B3]     {-1};
\node[]            (j0)     [below=0cm of B4]     {0};
\node[]            (j1)     [below=0cm of B5]     {1};
\node[]            (j2)     [below=0cm of B6]     {2};
\node[]            (j3)     [below=0cm of B7]     {3};
\node[]            (j4)     [below=0cm of B8]     {4};
\end{tikzpicture}
    \caption{Starting State of Embezzlement}
    \label{fig:emb_begin} 
\end{figure}
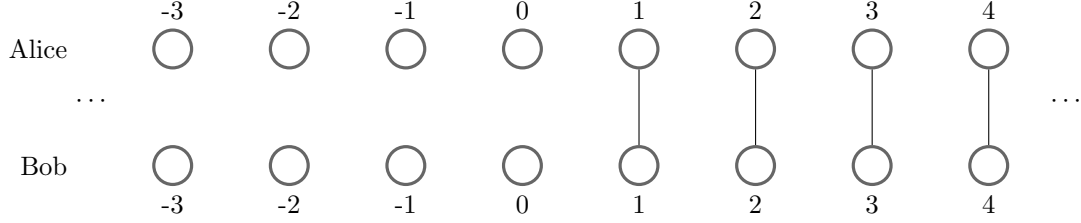

We begin with an infinite sequence of qubits indexed by the integers $\ldots, -2, -1, 0, 1, 2, \ldots$. The qubit pairs at positive indices are maximally entangled Bell states, indicated by lines connecting the qubits, while those at non-positive indices are initialized in the product state $\ket{0}$.

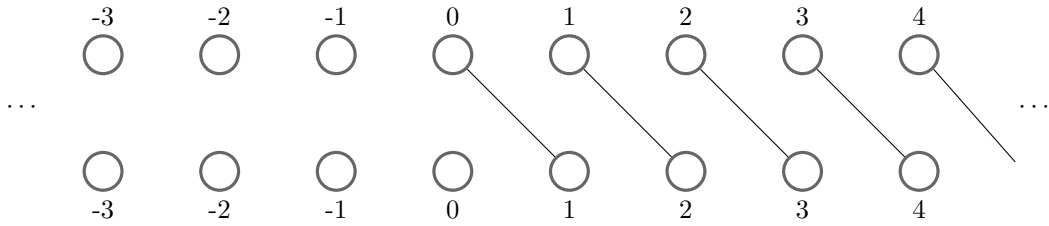
\begin{figure}[h!] 
    \centering
\begin{tikzpicture}[
qubitnode/.style={circle, draw=black!60, fill=white!5, very thick, minimum size=5mm},
]
\node[qubitnode]   (A1)                       {};
\node[qubitnode]   (A2)     [right=of A1]     {};
\node[qubitnode]   (A3)     [right=of A2]     {};
\node[qubitnode]   (A4)     [right=of A3]     {};
\node[qubitnode]   (A5)     [right=of A4]     {};
\node[qubitnode]   (A6)     [right=of A5]     {};
\node[qubitnode]   (A7)     [right=of A6]     {};
\node[qubitnode]   (A8)     [right=of A7]     {};
\node[below right=0.3cm and 1 cm of A8] {$\cdots$};
\node[below left=0.3cm and 0.5 cm of A1] {$\cdots$};

\node[qubitnode]   (B1)     [below=of A1]     {};
\node[qubitnode]   (B2)     [right=of B1]     {};
\node[qubitnode]   (B3)     [right=of B2]     {};
\node[qubitnode]   (B4)     [right=of B3]     {};
\node[qubitnode]   (B5)     [right=of B4]     {};
\node[qubitnode]   (B6)     [right=of B5]     {};
\node[qubitnode]   (B7)     [right=of B6]     {};
\node[qubitnode]   (B8)     [right=of B7]     {};
\node              (B9)      [right=of B8]     {$\ $};

\draw (A4.south east) -- (B5.north west);
\draw (A5.south east) -- (B6.north west);
\draw (A6.south east) -- (B7.north west);
\draw (A7.south east) -- (B8.north west);
\draw (A8.south east) -- (B9.north west);

\node[]            (i-3)    [above=0cm of A1]     {-3};
\node[]            (i-2)    [above=0cm of A2]     {-2};
\node[]            (i-1)    [above=0cm of A3]     {-1};
\node[]            (i0)     [above=0cm of A4]     {0};
\node[]            (i1)     [above=0cm of A5]     {1};
\node[]            (i2)     [above=0cm of A6]     {2};
\node[]            (i3)     [above=0cm of A7]     {3};
\node[]            (i4)     [above=0cm of A8]     {4};
\node[]            (j-3)    [below=0cm of B1]     {-3};
\node[]            (j-2)    [below=0cm of B2]     {-2};
\node[]            (j-1)    [below=0cm of B3]     {-1};
\node[]            (j0)     [below=0cm of B4]     {0};
\node[]            (j1)     [below=0cm of B5]     {1};
\node[]            (j2)     [below=0cm of B6]     {2};
\node[]            (j3)     [below=0cm of B7]     {3};
\node[]            (j4)     [below=0cm of B8]     {4};
\end{tikzpicture}
    \caption{Left shift of Alice's Qubits by 1}
    \label{fig:emb_Alice_left1} 
\end{figure}

To perform embezzlement, Alice and Bob each shift their respective qubits to the left by one position. Figure~\ref{fig:emb_Alice_left1} illustrates the intermediate state after Alice performs the shift alone, which results in a misalignment of the entangled pairs between Alice and Bob. Once Bob performs the corresponding left shift, the entanglement realigns, with the qubit pair at index 0 forming a Bell state.

\begin{figure}[h!]
    \centering
\begin{tikzpicture}[
qubitnode/.style={circle, draw=black!60, fill=white!5, very thick, minimum size=5mm},
]
\node[qubitnode]   (A1)                       {};
\node[qubitnode]   (A2)     [right=of A1]     {};
\node[qubitnode]   (A3)     [right=of A2]     {};
\node[qubitnode]   (A4)     [right=of A3]     {};
\node[qubitnode]   (A5)     [right=of A4]     {};
\node[qubitnode]   (A6)     [right=of A5]     {};
\node[qubitnode]   (A7)     [right=of A6]     {};
\node[qubitnode]   (A8)     [right=of A7]     {};
\node[below right=0.3cm and 0.5cm of A8] {$\cdots$};
\node[below left=0.3cm and 0.5cm of A1] {$\cdots$};
\node[qubitnode]   (B1)     [below=of A1]     {};
\node[qubitnode]   (B2)     [right=of B1]     {};
\node[qubitnode]   (B3)     [right=of B2]     {};
\node[qubitnode]   (B4)     [right=of B3]     {};
\node[qubitnode]   (B5)     [right=of B4]     {};
\node[qubitnode]   (B6)     [right=of B5]     {};
\node[qubitnode]   (B7)     [right=of B6]     {};
\node[qubitnode]   (B8)     [right=of B7]     {};
\node              (0)      [right=of B8]     {$\ $};
\draw (A4.south) -- (B4.north);
\draw (A5.south) -- (B5.north);
\draw (A6.south) -- (B6.north);
\draw (A7.south) -- (B7.north);
\draw (A8.south) -- (B8.north);

\node[]            (i-3)    [above=0cm of A1]     {-3};
\node[]            (i-2)    [above=0cm of A2]     {-2};
\node[]            (i-1)    [above=0cm of A3]     {-1};
\node[]            (i0)     [above=0cm of A4]     {0};
\node[]            (i1)     [above=0cm of A5]     {1};
\node[]            (i2)     [above=0cm of A6]     {2};
\node[]            (i3)     [above=0cm of A7]     {3};
\node[]            (i4)     [above=0cm of A8]     {4};
\node[]            (j-3)    [below=0cm of B1]     {-3};
\node[]            (j-2)    [below=0cm of B2]     {-2};
\node[]            (j-1)    [below=0cm of B3]     {-1};
\node[]            (j0)     [below=0cm of B4]     {0};
\node[]            (j1)     [below=0cm of B5]     {1};
\node[]            (j2)     [below=0cm of B6]     {2};
\node[]            (j3)     [below=0cm of B7]     {3};
\node[]            (j4)     [below=0cm of B8]     {4};

\node[qubitnode]   (Aq)     [above= of A5]        {};
\node[qubitnode]   (Bq)     [below= of B5]        {};

\draw[thick, <->] (A4.north east) -- (Aq.south west);
\draw[thick, <->] (B4.south east) -- (Bq.north west);

\end{tikzpicture}
    \caption{Swapping Out Qubits at Index 0}
    \label{fig:emb_swap} 
\end{figure}
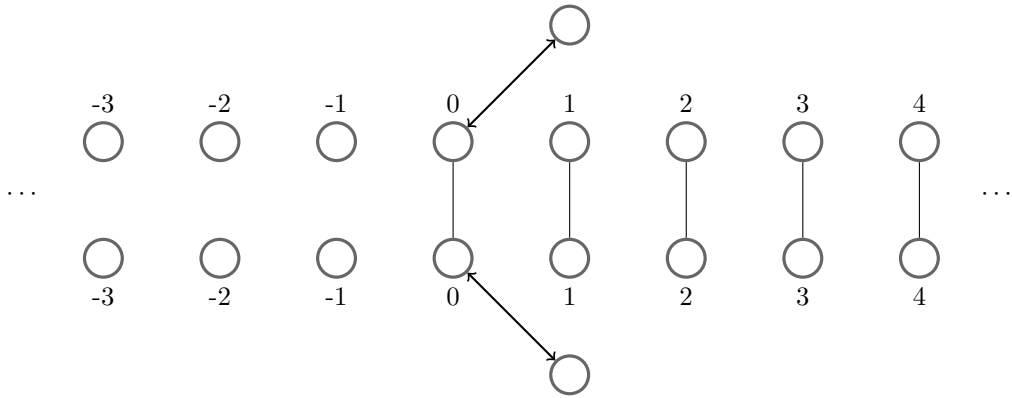

The final step is to swap out the qubit pair at index 0 with an external pair initialized in the $\ket{00}$ state. This swap restores the catalyst state to its original form (as in Figure~\ref{fig:emb_begin}) while producing a Bell pair outside the catalyst system. This completes the embezzlement process.

In \cite{cleve2017perfect}, the construction of the commuting operator protocol is somewhat technical because shifts or swaps performed unilaterally by Alice or Bob cause the resulting Hilbert space to become orthogonal to the original space. As a result, the protocol involves an infinite direct sum of infinite-dimensional Hilbert spaces. In addition, basis change between computational and Bell also needs to be performed at the right time and on right indices during the protocol, which adds further complexity to the construction.

 \subsection{Bell-state Embezzlement with CAR algebra}

We now define the C*-algebraic protocol for Bell-state embezzlement using the exact intuition above with the CAR algebra. 
The CAR algebra provides a simple and intuitive C*-algebraic model for systems with infinitely many qubits. Its elements behave like Pauli operators acting non-trivially on only finitely many sites, making it easy to visualize and work with. A key advantage is that the tensor product of two CAR algebras is again a CAR algebra, which gives a clean way to model local subsystems. This is in contrast to the commuting operator model, where local structure is less transparent and the tensor product is not canonical.
 More details about the CAR algebra can be found in the Appendix.

Recall that a CAR algebra \(\mathcal{R}\) can be viewed as a C*-algebra generated by infinitely many Pauli operators that act non-trivially on only finitely many qubits. Its generators are of the form \(X^a Z^b\), where \(a, b\) are infinite binary strings containing finitely many \(1\)'s, indicating the positions where the Pauli operators are not the identity. For example, \(X^{01} Z^{10} = Z \otimes X \otimes 1_\infty\).

We decompose the state in Figure~\ref{fig:emb_begin} into the \(\ket{00}\) and \(\ket{\psi} = \frac{\ket{00} + \ket{11}}{\sqrt{2}}\) components. Let \(\mathcal{R}\) be a CAR algebra, and define the infinite \(\ket{00}\) state as \(s_{\ket{00}}\) and the infinite Bell state as \(s_{\ket{\psi}}\).

Let \(s_{\ket{00}}: \mathcal{R} \otimes \mathcal{R} \to \mathbb{C}\) be an abstract state such that for any \(X^a Z^b, X^c Z^d \in \mathcal{R}\),
\begin{equation}
    s_{\ket{00}}(X^a Z^b \otimes X^c Z^d) = \prod_{i=0}^\infty \bra{00} X^{a_i} Z^{b_i} \otimes X^{c_i} Z^{d_i} \ket{00}. \label{eq:infinite_00_state}
\end{equation}
Here, \(a_i\) is the \(i\)-th bit of \(a\), and similarly for the others. The first CAR algebra corresponds to Alice's system, and the second corresponds to Bob's.

Similarly, let \(s_{\ket{\psi}}: \mathcal{R} \otimes \mathcal{R} \to \mathbb{C}\) be an abstract state such that for any \(X^a Z^b, X^c Z^d \in \mathcal{R}\),
\begin{equation}
    s_{\ket{\psi}}(X^a Z^b \otimes X^c Z^d) = \prod_{i=0}^\infty \bra{\psi} X^{a_i} Z^{b_i} \otimes X^{c_i} Z^{d_i} \ket{\psi}.
\end{equation}

Combining these two states, we define a state \(s: \mathcal{R} \otimes \mathcal{R} \to \mathbb{C}\) such that for any \(X^a Z^b, X^c Z^d \in \mathcal{R}\),
\begin{equation}
    s(X^a Z^b \otimes X^c Z^d) = \prod_{i=0}^{-\infty} \bra{00} X^{a_i} Z^{b_i} \otimes X^{c_i} Z^{d_i} \ket{00} \prod_{i=1}^\infty \bra{\psi} X^{a_i} Z^{b_i} \otimes X^{c_i} Z^{d_i} \ket{\psi}.
\end{equation}

Note that \(s = s_{\ket{00}} \otimes s_{\ket{\psi}}\), and we use the fact that \(\mathcal{R} \otimes \mathcal{R}\) is isomorphic to \(\mathcal{R}\). $a, b, c, d$ are now two-way infinite strings with finitely many 1's.

Next, we define local operators on Alice’s and Bob’s C*-algebras that correspond to shifting all qubits to the left. Let \(\Sigma\) be the set of all two-way infinite binary strings with finitely many \(1\)'s. Define \(\pi: \Sigma \to \Sigma\) by the shift operation: for all \(a \in \Sigma\), \(\pi(a)_i = a_{i-1}\). Physically, this means shifting all the qubits to the left by 1. Define the linear map \(\alpha_\pi\) on \(\mathcal{R}\) such that for all \(X^a Z^b \in \mathcal{R}\),
\begin{equation}
    \alpha_\pi(X^a Z^b) = X^{\pi(a)} Z^{\pi(b)}.
\end{equation}

We now verify that \(\alpha_\pi\) is a *-automorphism:
\begin{itemize}
    \item *-preservation:
    \begin{eqnarray}
        \alpha_\pi(X^a Z^b)^* = X^{\pi(a)} Z^{\pi(b)} = \alpha_\pi((X^a Z^b)^*).
    \end{eqnarray}
    \item Existence of inverse:
    \begin{equation}
        \alpha_\pi^{-1}(X^a Z^b) = X^{\pi^{-1}(a)} Z^{\pi^{-1}(b)} \quad \text{where} \quad \pi^{-1}(a)_i = a_{i+1}.
    \end{equation}
    \item Homomorphism:

    Since \(X\) and \(Z\) anticommute, we have
    \[
    X^a Z^b = (-1)^{f(a,b)} Z^{b} X^{a},
    \]
    where \(f: \Sigma \times \Sigma \to \{0,1\}\) is defined by
    \[
    f(a,b) = \bigoplus_{i=0}^\infty a_i \wedge b_i,
    \]
    with \(\wedge\) the binary AND and \(\oplus\) the binary XOR. It is clear that
    \[
    f(a,b) = f(\pi(a), \pi(b)) \quad \text{for all } a,b.
    \]
    Then for any \(X^a Z^b, X^c Z^d \in \mathcal{R}\),
    \begin{eqnarray}
        \alpha_\pi(X^a Z^b X^c Z^d) &=& (-1)^{f(b,c)} \alpha_\pi(X^{a \oplus c} Z^{b \oplus d}) \\
        &=& (-1)^{f(\pi(b), \pi(c))} X^{\pi(a) \oplus \pi(c)} Z^{\pi(b) \oplus \pi(d)} \\
        &=& (-1)^{f(\pi(b), \pi(c))} X^{\pi(a)} X^{\pi(c)} Z^{\pi(b)} Z^{\pi(d)} \\
        &=& X^{\pi(a)} Z^{\pi(b)} X^{\pi(c)} Z^{\pi(d)} \\
        &=& \alpha_\pi(X^a Z^b) \alpha_\pi(X^c Z^d).
    \end{eqnarray}
\end{itemize}

Next, we define the swap operation. Let \(\mathbf{a}, \mathbf{b}\) denote single-bit variables, distinguished from \(a, b \in \Sigma\). Define \(\alpha_{\mathrm{swap}}: \mathcal{R} \to \mathbb{M}_2\) as a linear map such that for any \(\mathbf{a}, \mathbf{b} \in \{0,1\}\) and \(a,b \in \Sigma\),
\begin{equation}
    \alpha_{\mathrm{swap}}(X^a Z^b \otimes X^{\mathbf{a}} Z^{\mathbf{b}}) = X^{a'} Z^{b'} \otimes X^{a_0} Z^{b_0},
\end{equation}
where
\[
a'_i = \begin{cases}
    \mathbf{a} & i=0 \\
    a_i & i \neq 0
\end{cases},
\]
and similarly for \(b'\).

It is clear that \(\alpha_{\mathrm{swap}}\) preserves the * operation and is its own inverse. The homomorphism property is straightforward, so we omit its proof.

Now define
\begin{equation}
\alpha = (\alpha_\pi \otimes \mathcal{I}) \circ \alpha_{\mathrm{swap}}. \label{eq:emb_automorphism}
\end{equation}
This \(*\)-automorphism acts on the C*-algebras, not the states, so the order of composition is reversed compared to the order of operations on the states. We can show that
\begin{equation}
    \alpha \otimes \alpha \bigl( X^a Z^b \otimes X^c Z^d \otimes X^{\mathbf{a}} Z^{\mathbf{b}} \otimes X^{\mathbf{c}} Z^{\mathbf{d}} \bigr)
    = X^{\pi(a')} Z^{\pi(b')} \otimes X^{\pi(c')} Z^{\pi(d')} \otimes X^{a_0} Z^{b_0} \otimes X^{c_0} Z^{d_0}.
\end{equation}
Therefore, when we apply the *-automorphisms to the catalyst state, we get
\begin{eqnarray*}
   && s \otimes s_0 \otimes s_0 \bigl(\alpha \otimes \alpha (X^a Z^b \otimes X^c Z^d \otimes X^{\mathbf{a}} Z^{\mathbf{b}} \otimes X^{\mathbf{c}} Z^{\mathbf{d}}) \bigr) \\
   &=& \prod_{i=0}^{-\infty} \bra{00} X^{a'_{i-1}} Z^{b'_{i-1}} \otimes X^{c'_{i-1}} Z^{d'_{i-1}} \ket{00}
   \prod_{i=1}^\infty \bra{\psi} X^{a'_{i-1}} Z^{b'_{i-1}} \otimes X^{c'_{i-1}} Z^{d'_{i-1}} \ket{\psi} \\
   && \times \bra{0} X^{a_0} Z^{b_0} \ket{0} \bra{0} X^{c_0} Z^{d_0} \ket{0} \\
   &=& \prod_{i=-1}^{-\infty} \bra{00} X^{a'_i} Z^{b'_{i-1}} \otimes X^{c'_i} Z^{d'_i} \ket{00}
   \prod_{i=0}^\infty \bra{\psi} X^{a'_i} Z^{b'_i} \otimes X^{c'_i} Z^{d'_i} \ket{\psi} \\
   && \times \bra{0} X^{a_0} Z^{b_0} \ket{0} \bra{0} X^{c_0} Z^{d_0} \ket{0} \\
   &=& \prod_{i=0}^{-\infty} \bra{00} X^{a_i} Z^{b_i} \otimes X^{c_i} Z^{d_i} \ket{00}
   \prod_{i=1}^\infty \bra{\psi} X^{a_i} Z^{b_i} \otimes X^{c_i} Z^{d_i} \ket{\psi} \\
   && \times \bra{\psi} X^{\mathbf{a}} Z^{\mathbf{b}} \otimes X^{\mathbf{c}} Z^{\mathbf{d}} \ket{\psi} \\
   &=& s \otimes s_{\phi} (X^a Z^b \otimes X^c Z^d \otimes X^{\mathbf{a}} Z^{\mathbf{b}} \otimes X^{\mathbf{c}} Z^{\mathbf{d}}),
\end{eqnarray*}
which results in the catalyst state tensor product with the Bell state, as expected. This completes the protocol of embezzlement in the C*-algebraic framework.

\section{Universal Embezzlement Protocol}
\subsection{Exact 2-qubit Embezzlement on a Dense Set}
With Bell-state embezzlement defined in the C*-algebraic model, we now have the basic building block for embezzling arbitrary 2-qubit states. The idea is simple: if we replace the catalyst with infinitely many copies of a fixed 2-qubit state $\ket\psi$, the same protocol lets us embezzle $\ket\psi$ exactly. To support more than one target state, we can build a larger catalyst by taking the tensor product of infinitely many such states. Now consider the family of states $\ket{\phi_q} = q_0\ket{00} + q_1\ket{11}$, where $q = |q_0|^2 / |q_1|^2 \in \mathbb Q$. This gives a countable set of entangled states that is dense in $\mathbb C^2\otimes \mathbb C^2$. Taking the tensor product over infinitely many copies of all $\ket{\phi_q}$ yields a catalyst that enables exact embezzlement of this dense subset, and approximate embezzlement of any 2-qubit state to arbitrary precision.

\begin{theorem}
There exists an explicit universal embezzlement protocol in the C*-model with a single catalyst state that enables exact embezzlement of a dense subset of bipartite states in $\mathbb C^2 \otimes \mathbb C^2$.
\end{theorem}
\begin{proof}

We start by defining the family of states we will use to build the catalyst state. Let $\ket{\psi_q} = q_0 \ket{00} + q_1\ket{11}$ for $q = |q_0|^2 / |q_1|^2 \in \mathbb Q^+$. We only need to consider cases where $q_1, q_1\in\mathbb R^+$ because any relative phase between the two coefficients can be achieved by a local unitary $U = \begin{pmatrix} 1 & 0\\ 0 & e^{i\gamma}\end{pmatrix}$ on Alice's side.  These states form a countable set that is dense in the set of entangled 2-qubit states (up to local unitaries), since the Schmidt coefficients of any entangled state can be approximated arbitrarily well by rational ratios.

\begin{equation}
    s_q(X^a Z^b\otimes X^c Z^d) = \prod_{i=0}^{-\infty} \bra{00} X^{a_i}Z^{b_i}\otimes X^{c_i}Z^{d_i}\ket{00}\prod_{i=1}^\infty \bra{\phi_q} X^{a_i} Z^{b_i} \otimes X^{c_i} Z^{d_i}\ket{\phi_q}
\end{equation}
where $a, b, c, d$ are again two-way infnite binary strings with finitey many 1's.

Since $\{q\}$ is a countable set, there is a one-to-one mapping between it and the set of natural numbers $\mathbb N$. We will inductively define Alice's and Bob's C*-algebras as countably infinite copies of CAR algebras using indexing from the natural numbers for clarity, and then changing the indexing to the set of positive rational numbers for our application.

Define
\begin{equation}
    \mathcal R_n := \bigotimes_{i=1}^n \mathcal R
\end{equation}
to be the tensor product of $n$ copies of $\mathcal R$ Let $\iota_n:\mathcal R_n \hookrightarrow \mathcal R_{n+1}$ be the natural embedding 
\begin{equation}
    \iota_n(a) := a\otimes 1.
\end{equation}
Now we can define the infinite tensor product of $\mathcal R$ as the inductive limit
\begin{equation}
    \mathcal R_\infty := \varinjlim \mathcal R_n  = \bigotimes_{i\in\mathbb N} \mathcal R.
\end{equation}
This inductive limit corresponds to the infinite tensor product of matrix algebras with finite support. We note that since $\mathcal R_\infty$ is inductively defined, any generator $a\in\mathcal R_\infty$ can be written in terms of $a_n\otimes 1$ for $a_n\in\mathcal R_n$ for some $n$, and therefore has only finite weight. $\mathcal R_\infty$ is also a CAR algebra.

Since $\mathbb Q^+$ is countable, we may relabel the inde set of the infinite tensor product via any bijection $\mathbb{N}\to\mathbb {Q}^+$ without changing the isomorphism class of the algebra. Now we swap the indexing from $\mathbb N$ to $\mathbb Q^+$, and define Alice and Bob's algebra to be 
\begin{equation}
   \mathcal A\otimes\mathcal B = \bigotimes_{q\in\mathbb Q^+} \mathcal R\otimes\mathcal R.
\end{equation}

Our catalyst state is then defined as $s_\psi:\mathcal A\otimes\mathcal B \to\mathbb C$
\begin{equation}
    s_\psi \left(\bigotimes_{q\in\mathbb Q^+}a_q\right)= \bigotimes_{q\in\mathbb Q^+} s_q (a_q)
\end{equation}
for $a_q\in\mathcal R\otimes\mathcal R$.

Now we have the catalyst state defined. The next step is to define the *-automorphism to embezzle the state $\ket{\phi_q}$.

Consider $\alpha_{q, \mathrm{swap}}:\mathcal R_{\infty}\otimes\mathbb M_2\to \mathcal R_{\infty}\otimes\mathbb M_2$ as the SWAP operation that swaps the qubit on index $(q, 0)$ of $\mathcal R_\infty$ with the outside qubit on $\mathbb M_2$. Consider $A \in\mathcal R_\infty$ where we write $A = \bigotimes_{q\in\mathbb Q^+} X^{a_q} Z^{b_q}$ where $a_q, b_q$ are infinite binary strings with finitely many 1's, and only finitely many $\{a_q, b_q\}_q$ are not all 0 strings. $\mathbf a$ and $\mathbf b$ are again a single bit corresponding to the outside qubit. Then
\begin{equation}
    \alpha_{q,\mathrm{swap}}(A\otimes X^{\mathbf a} Z^{\mathbf b}) = \bigotimes_{p\in\mathbb Q^+, \\ p\neq q} X^{a_p} Z^{b_p} \otimes\alpha_{\mathrm{swap}}(X^{a_q} Z^{b_q}\otimes X^{\mathbf a} Z^{\mathbf b})\label{eq:2-qubit_univ_swap}
\end{equation}
Next we define the shift operator on the $q$-th $\mathcal R$ of $\mathcal R_\infty$ in a similar manner. Let
$\alpha_{q,\pi}:\mathcal R_\infty\to\mathcal R_\infty$ be
\begin{equation}
    \alpha_{q, \pi}(A) = \bigotimes_{p\in\mathbb Q^+, p\neq q} X^{a_p}Z^{b_p} \otimes \alpha_\pi(X^{a_q} Z^{b_q}) \label{eq:2-qubit_univ_shift}
\end{equation}

In the expressions for $\alpha_{q, \mathrm{swap}}$ and $\alpha_{q, \pi}$, we isolate the $\mathcal R$ factor at index $q$ to make the action of the automorphisms explicit. This is purely a notational device -- the position of the algebra indexed by $q$ remains fixed in the overall tensor product, and no reordering of factors is involved.

We note that both the shift and the swap operators are just the shift and swap in the Bell-state embezzlement protocol tensor product with identity everywhere except on index $q$ of $\mathcal R_\infty$. They are therefore also *-automorphisms. Combining the shift and the swap, we get our overall *-automorphism 
\begin{equation}
    \alpha_q = \alpha_{q,\pi}\circ \alpha_{q,\mathrm{swap}}.
\end{equation}

The overall *-automorphism across Alice and Bob to embezzle state $\ket{\phi_q}$ is therefore $\alpha_q\otimes\alpha_q$. The operator only acts non-trivially on the register corresponding to state $s_q$. Therefore its action is exactly the single-state embezzlement protocol on state $s_q$ that produces $\ket{\phi_q}$ as the target state.

This completes the construction: for each $q \in \mathbb{Q}^+$, the *-automorphism $\alpha_q \otimes \alpha_q$ exactly embezzles the state $\ket{\phi_q}$ from the catalyst state $s_\psi$, without disturbing the rest of the system. Since the states $\ket{\phi_q}$ form a countable dense subset of all entangled two-qubit states, we have constructed a single exact embezzlement protocol and catalyst state capable of embezzling a dense set of states in $\mathbb{C}^2 \otimes \mathbb{C}^2$.

To approximately embezzle any state using this catalyst, we just need to pick a $\ket{\phi_q}$ that is $\epsilon$-close to the state we want to embezzle for the desired accuracy $\epsilon$, and perform the embezzlement protocol on it.
\end{proof}

\subsection{Exact $2n$-qubit Embezzlement on a Dense Set}
It is not difficult to extend the 2-qubit embezzlement protocol to the embezzlement of arbitrary $n$-qubit target states. All we need to do is include infinite copies of a dense set of $n$-qubit states in the catalyst, and then perform the embezzlement protocol on the appropriate component.

\begin{theorem}
There exists an explicit universal embezzlement protocol in the C*-algebraic model with a single catalyst state that enables exact embezzlement of a dense subset of bipartite states in $\mathbb C^{2^n} \otimes \mathbb C^{2^n}$ for any $n \in \mathbb N$.
\end{theorem}

\begin{proof}
A quantum system consisting of infinite copies of $n$-qubit states can also be modeled using CAR algebras. Consider $\mathbb M_{2^n}$, the algebra of $2^n \times 2^n$ matrices corresponding to observables on an $n$-qubit quantum system. Define
\begin{equation}
\mathcal M_k^{(n)} := \bigotimes_{i=1}^k \mathbb M_{2^n}
\end{equation}
to be the algebra of $k$ copies of $\mathbb M_{2^n}$, and define the unital embeddings
\begin{equation}
\iota_k: \mathcal M_k^{(n)} \hookrightarrow \mathcal M_{k+1}^{(n)}, \quad A \mapsto A \otimes \mathbb I_{2^n}.
\end{equation}
This yields an inductive system $(\mathcal M_k^{(n)}, \iota_k)$, and we define the inductive limit algebra as
\begin{equation}
\mathcal M^{(n)} := \varinjlim \mathcal M_k^{(n)}. \label{eq:car_n_infty}
\end{equation}
This algebra consists of infinite sequences of operators that act non-trivially on only finitely many qubits and as identity elsewhere. It can be regarded as the algebra of observables on infinitely many copies of an $n$-qubit system and is isomorphic to the CAR algebra.
    We now adapt the 2-qubit embezzlement protocol to the $2n$-qubit setting by applying the same mechanism to $\mathcal M^{(n)}$. The shift operation $\alpha_\pi$ shifts $n$ qubits to the left, and the swap operation $\alpha_{\mathrm{swap}}$ exchanges the $n$ qubits at index 0 with the external $n$-qubit state. The mathematical structure and derivation are essentially the same as in the 2-qubit case, with $X^a Z^b$ replaced by elements $M^a$ drawn from a generating set $\{M\} \subseteq \mathbb M_{2^n}$.

For universal embezzlement of $2n$-qubit states, we again take the catalyst to consist of infinite copies of a dense set of such states. This requires identifying a countably infinite dense subset.

Consider the Schmidt decomposition of a $2n$-qubit state with Schmidt coefficients $\{q_1, q_2, \ldots, q_{2^n}\}$. Define
\begin{equation}
    \ket{\phi_q} = \sum_{i=1}^{2^n} q_i \ket{ii},
\end{equation}
and restrict to $|q_i|^2 \in \mathbb Q^+$. Then the set of vectors $\vec q = (q_1, \ldots, q_{2^n})$ with rational squared amplitudes is a dense subset of the unit sphere in $\mathbb C^{2^n}$. Since each $q_i$ ranges over a countable set, the total collection of such $\vec q$ is also countable. Let us denote this set by $\mathbf Q_n$.

Define $s_q: \mathcal M^{(n)} \otimes \mathcal M^{(n)} \to \mathbb C$ to be the state corresponding to infinite copies of $\ket{\phi_q}$:
\begin{equation}
    s_q(M^a) := \prod_{i=1}^\infty \bra{\phi_q} M^{a_i} \ket{\phi_q}
\end{equation}
for all $M \in \mathbb M_{2^n}$ and where $a$ is a binary string with only finitely many 1s.

We now define the overall quantum system analogously to how we previously constructed $\mathcal R_\infty$, replacing $\mathcal R$ with $\mathcal M^{(n)}$, and reindexing by the set $\mathbf Q_n$. Define:
\begin{equation}
    \mathcal M_\infty^{(n)} := \bigotimes_{q \in \mathbf Q_n} \mathcal M^{(n)},
\end{equation}
and the catalyst state as:
\begin{equation}
    s_n := \bigotimes_{q \in \mathbf Q_n} s_q.
\end{equation}

Since $\mathcal M^{(n)}$ is isomorphic to the CAR algebra, the infinite tensor product $\mathcal M_\infty^{(n)}$ is also isomorphic to the CAR algebra.

The *-automorphisms that implement the embezzlement protocol are defined analogously to those in the 2-qubit case: they act as the $n$-qubit shift and swap operations on the $q$-th register and as the identity elsewhere. These automorphisms are constructed exactly as in Equations~\eqref{eq:2-qubit_univ_shift} and~\eqref{eq:2-qubit_univ_swap}, with $\mathcal R$ replaced by $\mathcal M^{(n)}$.

To perform approximate embezzlement of any $2n$-qubit state $\ket\phi$, locate the state $\ket{\phi_q}$ that is $\epsilon$-close to $\ket\phi$ for any $\epsilon>0$, and perform the embezzlement protocol on $\ket{\phi_q}$. This gives arbitrary approximate embezzlement protocol of any $2n$-qubit state.
\end{proof}

\subsection{Universal Embezzlement on a Dense Set}
    We now want to extend the system to allow embezzlement of arbitrary state of all dimension $n$ simultaneously.
    
\begin{theorem}
    There exists an explicit universal embezzlement protocol in the C*-model with a single catalyst state that enables exact embezzlement of a dense subset of bipartite state in $\mathbb C^{2^n}\otimes \mathbb C^{2^n}$ for all $n\in\mathbb N$.
\end{theorem}
\begin{proof}
    We do this by takin the tensor product of $\mathcal M_\infty^{(n)}$ for all $n$, again, inductively.

    Define $\mathcal M^N$ as
    \begin{equation}
        \mathcal M^N = \bigotimes_{i=1}^N \mathcal M^{(i)}_\infty
    \end{equation}
    with unital embedding
    \begin{equation}
        \iota_N:\mathcal M^N \hookrightarrow \mathcal M^{N+1},\ A \mapsto A\otimes 1.
    \end{equation}
    The indutive system $(\mathcal M^N, \iota_N)$ give rise to the inductive limit algebra as 
    \begin{equation}
        \mathcal M  = \varinjlim \mathcal M^N = \bigotimes_{i=1}^\infty \mathcal M^{(i)}_\infty
    \end{equation}
    which corresponds to Alice and Bob's local C*-algebra.
    
    The overall state is therefore 
   \begin{equation}
       s = \bigotimes_{i=1}^\infty s_i
   \end{equation} 
   To perform exact embezzlement of a $2n$-qubit state state $\ket{\phi_q}$, we first identify the algebra $\mathcal M^{(n)}$, then locate the index $q\in\mathbf Q_n$, and apply the corresponding single-state embezzlement protocol on the state $s_q$. The associated *-automorphism only acts non-trivially only on the register indexed by $q$ in $\mathcal M_\infty^{(n)}$, and as identity everywhere else.

   To perform approximate embezzlement of an arbitrary $2n$-qubit state $\ket\phi$, we locate the corresponding algebra $\mathcal M^{(n)}$, and choose a rational approximation $\ket{\phi_q}$ that is $\epsilon$-close to $\ket\phi$, and apply the embezzlement protocol associated with $s_q$.
\end{proof}
\begin{remark}
    Since each $\mathcal M_\infty^{(n)}$ is isomorphic to the CAR algebra, and the countable tensor product of cAR algebras remains isomorphic to the CAR algebra, the full system $\mathcal M$ is again isomorphic to the CAR algebra. Therefore, although the catalyst state now encodes dense subsets of states of all qubit dimensions simultaneously, the complexity of the underlying C*-algebra remains unchanged from the 2-qubit universal embezzlement system.
\end{remark}
\subsection{Exact Universal Embezzlement with Non-Separable System}\label{sec:non_separable}
Our final result establishes that exact universal embezzlement is possible for arbitary $2n$-qubit state, not merely a dense subset. To achieve this, we have to step into the realm of non-separable C*-algebras. 

\begin{theorem}
    With non-separable C*-algebra, there exists an explicit universal embezzlement protocol in the C*-model with a single catalyst state that can exactly embezzle any arbitray state with $2n$ qubits for all $n\in\mathbb N$.
\end{theorem}
\begin{proof}
    
The construction follows the same intuition as the dense set embezzlement protocol. Rather than taking a tensor product over a countable dense subset of states, we instead take a continuous tensor product over all the states, and apply the embezzlement protocol on the register corresponding to the target state.

    For clarity, we present the construction for exact universal embezzlement of arbitrary 2-qubit states; the generalization to \( 2n \)-qubit states proceeds analogously.

Let \( \ket{\psi_x} = x_0\ket{00} + x_1\ket{11} \) with \( x_0, x_1 \in \mathbb{R}^+ \) and \( x_1 \neq 0 \), and define the parameter \( x = x_0 / x_1 \). Define the state \( s_x : \mathcal{R} \otimes \mathcal{R} \to \mathbb{C} \) by
\[
    s_x (X^a Z^b \otimes X^c Z^d) = \prod_{i=0}^{-\infty} \langle 00 | X^{a_i} Z^{b_i} \otimes X^{c_i} Z^{d_i} | 00 \rangle \cdot \prod_{i=1}^{\infty} \langle \psi_x | X^{a_i} Z^{b_i} \otimes X^{c_i} Z^{d_i} | \psi_x \rangle.
\]

Let \( \mathcal{R}_x := \mathcal{R} \) denote a copy of the CAR algebra, indexed by \( x \in \mathbb{R}^+ \); we refer to each \( \mathcal{R}_x \) as a \emph{register}. Define the C*-algebra \( \mathcal{A} \) as the \textbf{finitely supported tensor product}
\[
    \mathcal{A} := \bigotimes_{x \in \mathbb{R}^+}^{\mathrm{fin}} \mathcal{R}_x,
\]
defined as the algebraic union
\[
    \bigotimes_{x \in \mathbb{R}^+}^{\mathrm{fin}} \mathcal{R}_x := \bigcup_{\substack{F \subset \mathbb{R}^+ \\ \text{finite}}} \bigotimes_{x \in F} \mathcal{R}_x.
\]
Elements of \( \mathcal{A} \) are norm-closed finite linear combinations of tensors \( \bigotimes_{x \in \mathbb{R}^+} a_x \) where:
\begin{itemize}
    \item \( a_x = 1 \) for all but finitely many \( x \in \mathbb{R}^+ \),
    \item each \( a_x \in \mathcal{R}_x = \mathcal{R} \),
    \item the norm is given by the C*-norm on the finite tensor product and extended by completion.
\end{itemize}

Define the catalyst state as
\[
    s := \bigotimes_{x \in \mathbb{R}^+}^{\mathrm{fin}} s_x.
\]
This state is well-defined on \( \mathcal{A} \), as elements of \( \mathcal{A} \) act non-trivially on only finitely many registers.

Although the algebra \( \mathcal{A} \) is non-separable, it is well-suited for our purposes: only finitely many registers are accessed at any time, and we need only act on the register corresponding to a single \( \mathcal R_x\) when performing embezzlement. The *-automorphisms used are exactly the same as those defined in the dense-set 2-qubit embezzlement protocol, and the same analysis shows that this yields an exact embezzlement protocol for any 2-qubit state.

The generalization to exact embezzlement of any \( 2n \)-qubit state follows by replacing \( \mathcal{R} \) with \( \mathcal{M}^{(n)} \), and defining the product over all normalized Schmidt vectors indexed appropriately.
\end{proof}

Even though the C*-algebra $\mathcal A$ is non-separable, nothing fundamental changes in the construction. All local operations in the protocol still act on finitely many registers, just like before. The *-automorphisms we use are the same as in the separable case -- they only touch the part of the system involved in the embezzlement and leave the rest unchanged. The catalyst state is still a tensor product of single-state catalysts, and the only difference is that now the index set is uncountable instead of countable.

In other words, the non-separability of $\mathcal A$ doesn't make the protocol any more complicated. We're not accessing or using any uncountably infinite part of the system at once -- every step still only needs to touch finitely many components. The larger algebra just gives us access to all possible target states at once, but the structure and mechanics of the protocol remain exactly the same.
\section{Discussions}
In a way, the construction can be made more succinct by factoring out the infinitely many $\ket{00}$ states in each component of the catalyst and replacing them with a single copy of infinitely many $\ket{00}$ states. However, this makes the mathematical expression of the protocol more cumbersome, and since everything is already infinite-dimensional, it does not actually reduce the overall complexity. We also note that there might exist more compact explicit embezzlement protocols that avoid taking the tensor product over all states, and instead rely on more sophisticated operations on a simpler catalyst. But since exact embezzlement necessarily requires infinite dimensions, any such “simpler” catalyst would still have to be infinite-dimensional, and therefore has the same essential complexity as our construction. We view the apparent redundancy of the catalyst state as a feature rather than a drawback, since having a catalyst built from explicit tensor products only makes the construction more transparent and intuitive.

We note that the catalyst state used in the dense-set 2-qubit embezzlement protocol is an infinite tensor product over states with rational Schmidt coefficients. When this state is passed through the GNS construction, it gives rise to the Type III$_1$ factor von Neumann algebra appearing in the Araki–Woods classification \cite{araki1968classification}. This aligns with the result of \cite{van2024embezzlement}, where the authors show that any universal embezzling state must give rise to a Type III$_1$ factor.

Our result on exact universal embezzlement for all states similarly requires a non-separable C*-algebra, mirroring the results in \cite{van2002embezzling}, where the authors mention the existence of exact universal embezzling states in non-separable Hilbert spaces.

\section{Acknowledgements}
The work was supported by the Novo Nordisk Foundation (grant NNF20OC0059939 ‘Quantum for Life’) and the European Research Council (grant agreement number 101078107 ‘QInteract)’ I would like thank Vern Paulsen for clarifying some of the mathematical concepts, and Richard Cleve for helpful discussions. I would also like to thank Ismael Profirio for pointing out that the non-separability for exact universal embezzlement has been proven in \cite{van2002embezzling}.

\bibliography{reference}

\begin{thebibliography}{10}

\bibitem{araki1968classification}
Huzihiro Araki and Woods EJ.
\newblock A classification of factors.
\newblock {\em Publications of the Research Institute for Mathematical Sciences, Kyoto University. Ser. A}, 4(1):51--130, 1968.

\bibitem{arveson1998invitation}
William Arveson.
\newblock {\em An invitation to C*-algebras}, volume~39.
\newblock Springer Science \& Business Media, 1998.

\bibitem{bratteli1982operator}
Ola Bratteli and Derek~W Robinson.
\newblock {\em Operator Algebras and Quantum Statistical Mechanics I}.
\newblock Springer, New York, 1982.

\bibitem{cleve2022constant}
Richard Cleve, Benoit Collins, Li~Liu, and Vern Paulsen.
\newblock Constant gap between conventional strategies and those based on c*-dynamics for self-embezzlement.
\newblock {\em Quantum}, 6:755, 2022.

\bibitem{cleve2017perfect}
Richard Cleve, Li~Liu, and Vern~I Paulsen.
\newblock Perfect embezzlement of entanglement.
\newblock {\em Journal of Mathematical Physics}, 58(1), 2017.

\bibitem{leung2013coherent}
Debbie Leung, Ben Toner, and John Watrous.
\newblock Coherent state exchange in multi-prover quantum interactive proof systems.
\newblock {\em Chicago Journal of Theoretical Computer Science}, 11(2013):1, 2013.

\bibitem{thesis}
Li~Liu.
\newblock Non-local quantum systems with infinite entanglement.
\newblock {\em Ph.D. Thesis, University of Waterloo}, 2022.

\bibitem{paulsen2002completely}
Vern~I. Paulsen.
\newblock {\em Completely Bounded Maps and Operator Algebras}.
\newblock Cambridge University Press, Cambridge, 2002.

\bibitem{Schwartzman:2024complexity}
Tal Schwartzman.
\newblock The complexity of entanglement embezzlement.
\newblock {\em arXiv preprint arXiv:2410.19051}, 2024.
\newblock Analyzes circuit complexity barriers to perfect embezzlement.

\bibitem{van2002embezzling}
Wim van Dam and Patrick Hayden.
\newblock Embezzling entangled quantum states.
\newblock {\em arXiv preprint quant-ph/0201041}, 2002.

\bibitem{van2024embezzlement}
Lauritz van Luijk, Alexander Stottmeister, Reinhard~F Werner, and Henrik Wilming.
\newblock Embezzlement of entanglement, quantum fields, and the classification of von neumann algebras.
\newblock {\em arXiv preprint arXiv:2401.07299}, 2024.

\bibitem{vanLuijk:2024multipartite}
Lauritz van Luijk, Alexander Stottmeister, and Henrik Wilming.
\newblock Multipartite embezzlement of entanglement.
\newblock {\em arXiv preprint arXiv:2409.07646}, 2024.
\newblock Finite-dimensional approximations and multipartite universal embezzlement.

\bibitem{zanoni2024complete}
Elia Zanoni, Thomas Theurer, and Gilad Gour.
\newblock Complete characterization of entanglement embezzlement.
\newblock {\em Quantum}, 8:1368, 2024.

\end{thebibliography}
\bibliographystyle{plain}

\appendix
\section*{Appendix A: Introduction to C*-algebraic Model}

C*-algebras provide a natural and flexible framework to describe quantum systems algebraically. At their core, they generalize algebras of bounded operators on Hilbert spaces but do so in a fully abstract, coordinate-free manner. This abstraction makes them extremely useful in infinite-dimensional quantum theory, quantum statistical mechanics, and beyond.

We look at a detailed overview of C*-algebras, states, the GNS construction, and tensor product constructions. These notions form the algebraic backbone for modeling quantum systems in infinite dimensions and are essential for understanding the operator-algebraic framework utilized in this work. Our exposition follows standard references such as \cite{arveson1998invitation, bratteli1982operator, paulsen2002completely}.

\setcounter{section}{1}
\subsection{Basics of C*-Algebras}
\begin{definition}
A \emph{C*-algebra} is a Banach algebra \(\mathcal{A}\) over the complex numbers equipped with an involution \( * : \mathcal{A} \to \mathcal{A} \) satisfying the \emph{C*-identity}:
\[
\| a^* a \| = \| a \|^2 \quad \text{for all } a \in \mathcal{A}.
\]
More concretely, \(\mathcal{A}\) is a complex algebra with norm \(\|\cdot\|\) and a conjugate-linear involution \(a \mapsto a^*\) such that for all \(a,b \in \mathcal{A}\) and \(\lambda, \mu \in \mathbb{C}\):
\begin{itemize}
    \item \((a^*)^* = a\),
    \item \((ab)^* = b^* a^*\),
    \item \((\lambda a + \mu b)^* = \overline{\lambda} a^* + \overline{\mu} b^*\),
    \item \(\mathcal{A}\) is complete with respect to \(\|\cdot\|\).
\end{itemize}
\end{definition}

Intuitively, a C*-algebra is an abstract algebraic structure that generalizes the algebra of bounded linear operators on a Hilbert space, capturing both algebraic and analytic properties of quantum observables \cite[Ch.~1]{arveson1998invitation}.

\begin{definition}
A \emph{*-homomorphism} \(\phi : \mathcal{A} \to \mathcal{B}\) between two C*-algebras is a linear map preserving the algebraic operations and the involution:
\[
\phi(ab) = \phi(a) \phi(b), \quad \phi(a^*) = \phi(a)^*,
\]
for all \(a,b \in \mathcal{A}\).
\end{definition}

If \(\mathcal{A}\) and \(\mathcal{B}\) are unital, *-homomorphisms also satisfy \(\phi(1_{\mathcal{A}}) = 1_{\mathcal{B}}\). Such maps are automatically norm-contracting \cite[Prop.~1.3.6]{bratteli1982operator}, i.e., \(\|\phi(a)\| \leq \|a\|\), ensuring continuity.

A *-homomorphism that is bijective is called a *-isomorphism, and a *-isomorphism from \(\mathcal{A}\) to itself is a *-automorphism. In our C*-algebraic formalism for quantum information, *-automorphisms describe the dynamics of quantum systems via transformations on the set of observables.

\begin{definition}
An \emph{abstract state} \(\omega\) on a unital C*-algebra \(\mathcal{A}\) is a positive linear functional normalized by \(\omega(1) = 1\). That is,
\[
\omega(a^* a) \geq 0 \quad \text{and} \quad \omega(1) = 1.
\]
\end{definition}

States generalize the notion of density operators (quantum states) in Hilbert spaces to the algebraic setting, providing expectation values for observables \(a \in \mathcal{A}\). The set of all states on \(\mathcal{A}\) forms a convex, weak*-compact subset of the dual space \(\mathcal{A}^*\) \cite[Ch.~4]{paulsen2002completely}. To connect an abstract state \(\omega\) with a quantum state \(\ket\psi\) in a Hilbert space \(\mathcal{H}\), we write \(\omega(a) = \bra\psi \pi(a) \ket\psi\), where \(\pi(a)\) is the Hilbert space representation of \(a\).

\subsection{The GNS Construction}

The Gelfand–Naimark–Segal (GNS) construction provides a canonical way to represent a C*-algebra \(\mathcal{A}\) as bounded operators on a Hilbert space, starting from a state \(\omega\). This bridges the abstract algebraic formulation with the Hilbert space formalism of quantum mechanics.

\begin{theorem}[GNS Construction]
Given a state \(\omega\) on \(\mathcal{A}\), there exists a triple \((\pi_\omega, \mathcal{H}_\omega, \Omega_\omega)\) where:
\begin{itemize}
    \item \(\mathcal{H}_\omega\) is the Hilbert space obtained by completing the quotient \(\mathcal{A} / \mathcal{N}_\omega\), with inner product \(\langle [a], [b] \rangle = \omega(b^* a)\),
    \item \(\pi_\omega : \mathcal{A} \to B(\mathcal{H}_\omega)\) is a *-representation defined by \(\pi_\omega(a)[b] = [ab]\),
    \item \(\Omega_\omega = [1]\) is a cyclic vector satisfying \(\omega(a) = \langle \Omega_\omega, \pi_\omega(a) \Omega_\omega \rangle\),
\end{itemize}
where \(\mathcal{N}_\omega = \{ a \in \mathcal{A} : \omega(a^* a) = 0 \}\) is a left ideal.
\end{theorem}

This construction realizes any abstract state as a vector state in a Hilbert space representation. Importantly, when the C*-algebra is constructed as a tensor product of local algebras, the GNS representation yields commuting local operator algebras on the same Hilbert space, matching the commuting operator framework. A detailed discussion of this equivalence can be found in \cite[Appendix~C]{cleve2022constant}.

\subsection{Tensor Products of C*-Algebras}

Given two C*-algebras \(\mathcal{A}\) and \(\mathcal{B}\), their algebraic tensor product \(\mathcal{A} \otimes_{\mathrm{alg}} \mathcal{B}\) can be completed in multiple ways depending on the norm used. The most important completions are:

\begin{itemize}
    \item \textbf{Minimal (spatial) tensor product} \(\mathcal{A} \otimes_{\min} \mathcal{B}\): The minimal (spatial) tensor product is defined as
    \begin{equation}
    \|x\|_{\min} = \sup\{\|\pi_1\otimes \pi_2(x)\|:\pi_1:\mathcal A\to\mathbb B(\mathcal H_1), \pi_2: \mathcal B\to\mathbb B(\mathcal H_2) \text{ are unital *-homomorphisms}\}    
    \end{equation}
    This is the physically relevant tensor product for describing non-interacting quantum systems with separate subsystems \cite{bratteli1982operator}.

    \item \textbf{Maximal tensor product} \(\mathcal{A} \otimes_{\max} \mathcal{B}\): Defined by the maximal C*-norm:
    \[
    \|x\|_{\max} = \sup \{ \| \pi(x) \| : \pi: \mathcal{A} \otimes_{\mathrm{alg}} \mathcal{B}\to \mathbb B(\mathcal H) \text{ is a unital *-homomorphism} \},
    \]
    where the supremum is taken over all *-representations (not necessarily spatial). This yields the most general composite system and may encode stronger-than-quantum correlations \cite{bratteli1982operator}.
\end{itemize}

These tensor products generally differ unless one of the algebras is \emph{nuclear}, in which case all reasonable C*-tensor norms coincide \cite{arveson1998invitation}.

In general, we use the tensor product of C*-algebras to model local quantum systems. This formalism has the advantage of providing a canonical way to separate different local subsystems via tensor products. States and operators on individual subsystems are easy to define and can be combined naturally using the algebraic tensor product. Interestingly, the C*-algebraic model can be shown to be equivalent to the commuting operator model through the GNS construction. A more detailed discussion of this equivalence can be found in the appendix of \cite{cleve2022constant}. This equivalence enables us to use the C*-algebraic approach to meaningfully define entanglement and non-local quantum systems even in the infinite-dimensional setting.

\section*{Appendix B: The CAR algebra}
The CAR (Canonical Anticommutation Relations) algebra is a fundamental object in studying infinite dimensional quantum systems, especially those involving infinitely many qubits similar to the ones in the fermionic chain. In this appendix, we give the standard definition of the CAR algebra and explain the properties that are relevant to the constructions in our paper.
\begin{definition}[CAR Algebra]
Let $\mathcal H$ be a separable Hilbert space. The CAR algebra over $\mathcal H$, denoted by $\mathrm{CAR}(\mathcal H)$, is the C*-algebra generated by elements $a(f)$ for $f \in \mathcal H$, satisfying the canonical anticommutation relations:
\begin{equation}
a(f) a(g) + a(g) a(f) = 0,\quad a(f)a(g)^* + a(g)^* a(f) = \langle f, g\rangle 1,
\end{equation}
where $\langle f, g \rangle$ is the inner product between $f$ and $g$.
The algebra $\mathrm{CAR}(\mathcal H)$ is uniquely determined (up to *-isomorphism) by these relations.
\end{definition}
\begin{example}[CAR on Single Qubit]
Take $\mathcal H = \mathbb C$. Then the CAR algebra coincides with $\mathbb M_2$, as shown below. We have a single generator $a$ satisfying $\{a, a^*\} = 1$. We can write:
\begin{equation}
a = \begin{pmatrix} 0 & 1\\ 0 & 0 \end{pmatrix},\quad
a^* = \begin{pmatrix} 0 & 0 \\ 1 & 0 \end{pmatrix} \label{eq:car_c1_a}
\end{equation}
This operator can also be expressed in terms of Pauli matrices:
\begin{equation}
a = \frac{1}{2}(X + iY),\quad a^* = \frac{1}{2}(X - iY).
\end{equation}
This algebra therefore corresponds to the algebra of observables for a single-qubit quantum system.
\end{example}
\begin{example}[CAR on $n$ Qubits]
We now consider the algebra on $n$ qubits. Taking $n$ copies of $\mathbb M_2$, the observable algebra is given by:
\begin{equation}
\mathcal R_n = \bigotimes_{i=1}^n \mathbb M_2.
\end{equation}
We can view $\mathcal R_n$ as the CAR algebra over $\mathbb C^n$, i.e. $\mathrm{CAR}(\mathbb C^n) \cong \mathbb M_{2^n}$.
We show this inductively. Let $\mathcal A_1 = \mathrm{CAR}(\mathbb C)$. We define:
\begin{equation}
\mathcal A_{n+1} \cong \mathcal A_n \otimes \mathcal A_1,
\end{equation}
with the embedding:
\begin{equation}
a_i \mapsto a_i \otimes I \quad (\forall i \leq n),\qquad a_{n+1} \mapsto \Gamma \otimes a,
\end{equation}
where:
\begin{itemize}
\item $a$ is the annihilation operator from Equation~\ref{eq:car_c1_a},
\item $\Gamma$ is the parity operator on $\mathcal A_n$ defined as:
\begin{equation}
\Gamma = (-1)^N,\quad N = \sum_{i=1}^n a_i^* a_i,
\end{equation}
where $(-1)^N$ is defined by spectral calculus: if $N \ket{\psi} = n \ket{\psi}$, then $(-1)^N \ket{\psi} = (-1)^n \ket{\psi}$.
\end{itemize}
This ensures:
\begin{equation}
\{ \Gamma \otimes a,\ a_i \otimes I \} = 0,
\end{equation}
so $a_{n+1}$ anticommutes with all previous generators, and we obtain $\mathrm{CAR}(\mathbb C^{n+1}) \cong \mathbb M_{2^{n+1}}$.
\end{example}
\begin{example}[Infinite Qubit CAR Algebra]
To describe infinitely many qubits, we take the inductive limit of $\mathcal R_n$ using the embedding $\iota_n: \mathcal R_n \hookrightarrow \mathcal R_{n+1}$ defined by $A \mapsto A \otimes I_2$. This defines:
\begin{equation}
\mathcal R = \varinjlim \mathcal R_n = \bigotimes_{i=1}^\infty \mathbb M_2.
\end{equation}
This infinite tensor product is the CAR algebra used in our construction. Unless otherwise noted, references to the CAR algebra refer to this $\mathcal R$.
\end{example}
\begin{remark}
We also consider systems where each component is a copy of an $n$-qubit system, i.e., $\mathbb M_{2^n}$, and take the infinite tensor product:
\begin{equation}
\mathcal M^{(n)} = \bigotimes_{i=1}^\infty \mathbb M_{2^n}.
\end{equation}
Since $\mathbb M_{2^n} \cong \mathrm{CAR}(\mathbb C^n)$ and $\mathrm{CAR}(\mathbb C^n) \cong \bigotimes_{i=1}^n \mathbb M_2$, it follows that:
\begin{equation}
\varinjlim \mathbb M_{2^n} \cong \varinjlim (\mathbb M_2^{\otimes n}) \cong \varinjlim \mathbb M_2,
\end{equation}
so $\mathcal M^{(n)}$ is again isomorphic to the same infinite qubit CAR algebra $\mathcal R$.
\end{remark}
\begin{theorem}[Tensor Product of CAR Algebras]
Let $\mathcal R$ be the CAR algebra. Then $\mathcal R \otimes \mathcal R \cong \mathcal R$, and the minimal and maximal tensor products coincide:
\begin{equation}
\mathcal R \otimes_{\min} \mathcal R = \mathcal R \otimes_{\max} \mathcal R.
\end{equation}
\end{theorem}
\begin{proof}
We observe:
\begin{equation}
\mathcal R \otimes \mathcal R \cong \left( \bigotimes_{i=1}^\infty \mathbb M_2 \right) \otimes \left( \bigotimes_{i=1}^\infty \mathbb M_2 \right) \cong \bigotimes_{i=1}^\infty \mathbb M_2,
\end{equation}
up to relabeling of tensor indices.
To see that the minimal and maximal tensor products coincide, we note that $\mathbb M_2$ is nuclear, and nuclearity is preserved under inductive limits. \cite[Thm 6.3.6]{bratteli1982operator} Since each $\mathbb M_2$ is nuclear and the CAR algebra is built as an inductive limit of such algebras, $\mathcal R$ is nuclear, and thus its tensor product is unambiguous.
\end{proof}
\begin{remark}
In Section~\ref{sec:non_separable}, we defined an uncountable tensor product of CAR algebras with finite support. The same argument applies: since each component is nuclear and the tensor product has finite support, the minimal and maximal C*-tensor products agree. The norm is unambiguous in this case.
\end{remark}

\end{document}